\documentclass[floatfix, showpacs, preprintnumbers, twocolumn, amsmath, amssymb, aps]{revtex4}

\usepackage{bm}
\usepackage{graphicx}
\usepackage{bbold}
\usepackage{url}
\usepackage{subfigure}

\newtheorem{theorem}{Theorem}[section]
\newtheorem{lemma}{Lemma}[section]
\newenvironment{proof}[1][Proof:]{\begin{trivlist} \item[\hskip \labelsep {\bfseries #1}]}{\end{trivlist}}

\newcommand{\bra}[1]{\left< #1 \right|} 
\newcommand{\ket}[1]{\left| #1 \right>}

\begin{document}

\title{Extending matchgates into universal quantum computation}

\author{Daniel J. Brod}
\email{brod@if.uff.br}
\affiliation{Instituto de F\'isica, Universidade Federal Fluminense, Av. Gal. Milton Tavares de Souza s/n, \\
Gragoat\'a, Niter\'oi, RJ, 24210-340, Brazil}

\author{Ernesto F. Galv\~ao}
\email{ernesto@if.uff.br}
\affiliation{Instituto de F\'isica, Universidade Federal Fluminense, Av. Gal. Milton Tavares de Souza s/n, \\
Gragoat\'a, Niter\'oi, RJ, 24210-340, Brazil}

\date{\today}

\begin{abstract}
Matchgates are a family of two-qubit gates associated with noninteracting fermions. They are classically simulatable if acting only on nearest neighbors, but become universal for quantum computation if we relax this restriction or use SWAP gates [Jozsa and Miyake, Proc. R. Soc. A \textbf{464}, 3089 (2008)]. We generalize this result by proving that any nonmatchgate parity-preserving unitary is capable of extending the computational power of matchgates into universal quantum computation. We identify the single local invariant of parity-preserving unitaries responsible for this, and discuss related results in the context of fermionic systems.

\end{abstract}
\maketitle

\section{Introduction}

There are good reasons to study quantum circuits with limited allowed operations. Experimentally, they may correspond to circuits implementable with current technologies or using a particular physical system. Theoretically, they can help us understand the minimal requirements for the quantum computational speedup achieved in problems such as factoring. When the limited operations are not sufficient for universal quantum computation, identifying the missing resource can be revealing.  As an example, the set of single-qubit gates operating on a separable initial state can be simulated classically, but the addition of any entangling gate to the set makes it universal for quantum computation.

Matchgates are a restricted family of two-qubit gates first described by Valiant \cite{Valiant02} in the context of graph theory. Valiant showed that circuits composed only of such gates acting on nearest-neighbor qubits were classically simulatable. Nearest-neighbor matchgates were shown to be equivalent to noninteracting fermions \cite{Terhal02}, and remain classically simulatable even in the case of adaptive measurements. This is in contrast to noninteracting bosons, which allow for universal quantum computation with adaptive measurements \cite{KLM01}. 

Many results were obtained by studying the algebraic properties of matchgates, with no explicit reference to their fermionic nature. For instance, it was shown that the computational power of nearest-neighbor matchgates is equivalent to that of space-bounded quantum computation \cite{Jozsa10} and of linear threshold gates \cite{Nest10}. It was also shown that a universal set can be obtained by relaxing slightly the nearest-neighbor condition, allowing matchgates to act on next-nearest-neighbor qubits (i.e., one qubit apart) or, equivalently, by adding the SWAP gate to the set \cite{Jozsa08b}. This algebraic study of matchgates is important not only for understanding fundamental properties of quantum circuits, but also to study possible experimental implementations in different physical systems such as linear optics, as proposed in \cite{Ramelow10}.

In this paper we address the question of which resources are sufficient to extend the computational power of matchgates into full quantum computation. We do this by studying parity-preserving two-qubit gates, a natural generalization of matchgates. We expand on the result of \cite{Jozsa08b} by showing that any (nonmatchgate) member of this larger set, together with matchgates, is sufficient to attain universal quantum computation. Our discussion relies on the characterization of two-qubit gates in terms of their nonlocal parameters \cite{Makhlin02, Zhang03}. These are quantities invariant by local, single-qubit unitaries; one example is the entangling power defined in \cite{Zanardi00}. We show that a parity-preserving gate's ability to extend matchgates into quantum universality can be attributed to a particular nonlocal parameter. This identifies the common characteristic shared by gates as different as SWAP and CZ which enables them to boost the computational power of matchgates \cite{Jozsa08b, Bravyi02}.

The paper is structured as follows. In section \ref{sec pp} we define matchgates and parity-preserving gates, recalling important known results and demonstrating the universality of matchgates + SWAP in a simplified way, and in section \ref{sec nlpar} we review the characterization of the nonlocal parameters of two-qubit gates. In section \ref{sec iswap}, we suggest an intuitive generalization of the SWAP into a continuous family of two-qubit gates, which can provably boost the computational power of matchgates. In section \ref{sec matchnlpar}, we provide a complete characterization of matchgates and parity preserving unitaries in terms of their nonlocal parameters, necessary to obtain our main result in section \ref{sec main}. In that section we show that any nonmatchgate parity-preserving gate forms a universal set together with matchgates. Sections \ref{sec disc} and \ref{sec conclusion} contain a discussion of the results, open problems and conclusions.

\section{Review} \label{sec review}

We start by reviewing some basic concepts, definitions and known results which will be necessary for our main discussion. In section \ref{sec pp} we recall the definition of two-qubit matchgates and show how they are inserted in the larger group of parity-preserving matrices, referring to known results regarding simulability and universality of these sets. In section \ref{sec nlpar}, we review the characterization of the nonlocal parameters of general two-qubit gates.

\subsection{Parity-preserving unitaries and matchgates} \label{sec pp}

Let us consider the most general 4x4 unitary matrix that preserves the parity subspaces ($\{ \ket{00},\ket{11} \}$ and $\{ \ket{10},\ket{01} \}$) of a two-qubit system. Denoting it by $G(A,B)$, where $A$ and $B$ are the 2x2 matrices acting on the even and odd subspaces, respectively, we can write it as
\begin{equation} \label{PPgates}
G(A,B) = \left(\begin{array}{cccc}
A_{11} & 0 & 0 & A_{12} \\
0 & B_{11} & B_{12} & 0 \\
0 & B_{21} & B_{22} & 0 \\
A_{21} & 0 & 0 & A_{22}
\end{array}\right).
\end{equation}

Since $G(A,B) \cdot G(C,D)=G(AC,BD)$, the set of all unitary matrices of this type forms a subgroup of U(4). From this product rule and the unitarity of $G(A,B)$ also comes the fact that $A$ and $B$ must be unitary as well. With no loss of generality, we can fix the global phase of $G(A,B)$ by imposing that det($A$)det($B$)$ = 1$. If we also impose the nontrivial constraint that det($A$)=det($B$)= $\pm$ 1, the resulting unitary is known as a \textit{matchgate} \footnote{Sometimes, to maintain common convention, we will not take $G(A,B)$ to belong to SU(4). In this case, $G(A,B)$ will be a matchgate iff det($A$)=det($B$).}. From this point on, parity-preserving (P.P., for short) denotes general unitaries of the form $G(A,B)$, while \textit{matchgate} denotes the special case where the additional determinant constraint is satisfied. Sometimes, we will also refer to ``nonmatchgate P.P. unitaries'', when det($A) \neq $ det($B$).

As mentioned before, the computational power of matchgates has been studied in the literature both in the context of noninteracting fermions and of the algebraic properties of the matrices themselves. The most well-known result is due to Valiant \cite{Valiant02}, who defined matchgates in the context of graph theory and proved that circuits composed only of nearest-neighbor matchgates can be classically simulated. This result was reproduced in the context of fermions by Terhal and DiVincenzo \cite{Terhal02}, who showed that circuits composed only of nearest-neighbor matchgates were mathematically equivalent to systems of noninteracting fermions. They also showed that the simulability result holds even in the presence of adaptive measurements, in contrast to noninteracting bosons, for which adaptive measurements allow for universal quantum computation \cite{KLM01}. Subsequently, several results were obtained regarding how to increment this set of operations to attain universal quantum computation, for example by adding certain simple interactions  \cite{Bravyi02} or allowing for charge measurements \cite{Beenakker04}.

In this work, we are interested in the algebraic properties of matchgates. In particular, we address the question of which unitaries can single-handedly boost the computational power of (nearest-neighbor) matchgates to that of full universal quantum computation. One such example is the SWAP gate, since it replaces the need for long-range interactions, and it is known \cite{Terhal02, Kempe01b} that matchgates acting on arbitrary pairs of qubits form a universal set. Even more, it has been shown in \cite{Jozsa08b} that one does not need to fully remove the nearest-neighbor constraint, and only relaxing it slightly to allow for next-nearest-neighbor interactions (i.e., between qubits one line apart) is sufficient to achieve (encoded) universality. 

Since our work is heavily based on the construction proposed in \cite{Jozsa08b}, we now reproduce a simplified demonstration of this result. In \cite{Jozsa08b}, each logical qubit is encoded in four physical qubits. Here, we will use an encoding which is more economical (each logical qubit is encoded in two physical qubits), and which is equivalent up to some technical reasons which need not concern us. We encode each logical qubit in the even parity subspace of two physical qubits:
\begin{equation} \label{encoding}
 \begin{split}
  \ket{0}_L & = \ket{00}, \\
  \ket{1}_L & = \ket{11}.
 \end{split}
\end{equation}

Since matchgates act separately on parity subspaces, any single-qubit gate $A$ is implemented in the encoded qubit simply by the matchgate $G(A,A)$, as shown in Fig. \ref{JozsaDem}a.

\begin{figure}
     \subfigure[]{
\includegraphics[width=1.8in]{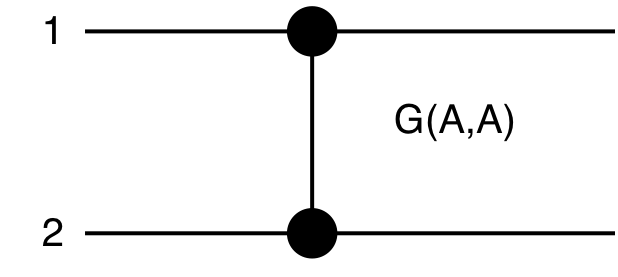}}
 
     \subfigure[]{
\includegraphics[width=2.3in]{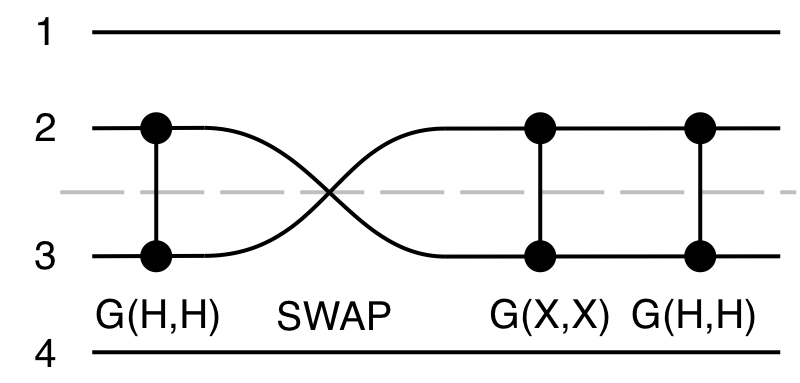}}
     
\caption{Universal set of gates in terms of matchgates and SWAP: (a) Arbitrary single-qubit gate A and (b) CZ gate.}
\label{JozsaDem}
\end{figure} 

Besides arbitrary single-qubit gates, for universality it is sufficient to implement one of the standard maximally entangling gates. One choice specially suitable for this construction is the controlled-$Z$ [$CZ =$ diag$(1, 1, 1, -1)$], depicted in Fig. \ref{JozsaDem}b. Let us label the physical qubits that encode the first logical qubit with numbers 1 and 2, and similarly with the second logical qubit (physical qubits 3 and 4). In view of the encoding, it is simple to see that a CZ between the logical qubits can be implemented by a single CZ between physical qubits 2 and 3 only. This is the simplest choice which preserves the encoding while still avoiding the need for long-range interactions (compare with the CNOT gate, for example, which would have to act on qubits 3 and 4 controlled by the state of either qubit 1 or 2).

It can be easily checked that the circuit in Fig. \ref{JozsaDem}b, implementing
\begin{equation}
CZ_{23} = G(H,H) G(X,X) \; \textrm{SWAP} \; G(H,H),
\end{equation}
has the desired effect and thus completes the universal set (in the encoded space), having used only matchgates and the SWAP gate. Note that qubit initialization and measurements are not an issue here, as both preparation and measurement in the logical computational basis are trivially implemented by preparation and measurement in the physical computational basis, and there is no need for previously entangled or otherwise inaccessible states. In short, this gate-set can efficiently simulate any quantum computation. 

There are other examples in quantum information of a single gate uplifting a classically simulatable set into universal quantum computation. For instance, circuits with nonentangling gates only are classically simulatable, while adding a single entangling gate allows for full quantum computation. Another example is the Toffoli gate, which by itself results in classically simulatable circuits (in fact, Toffoli is universal for classical computation), but adding the Hadamard gate promotes the set to quantum universality \cite{Shi02}. However, the matchgate + SWAP result is also surprising as the added resource seems to be a very trivial one (the simple swapping of qubits), as opposed to the resources in the other cases, which are considered to be fundamentally quantum in nature (entanglement and quantum superposition, respectively). Also of notice is that, in contrast with the Toffoli gate, matchgates are expected to be much weaker even than full \textit{classical} computation, as it can be shown that any boolean function calculated by a circuit composed of nearest-neighbor matchgates is trivial, in the sense that it depends only on one bit of the input \cite{Nest10}. In this sense, the SWAP gate is bridging the gap between a sub-classical complexity class and full quantum computing.

In \cite{Jozsa08b}, the role of the SWAP gate was seen as to provide the long-range interactions missing in nearest-neighbor matchgates. In this case, we could expect other gates that exchange the role of the qubits in the computational basis to do just as well, such as iSWAP [$ = G(I,iX)$], and fermionic SWAP [$ = G(Z,X)$]. It is easy to check, however, that these two gates are matchgates, and thus cannot confer universality to matchgates as SWAP can. Another aspect in which they differ from the SWAP is that they are entangling. The fermionic SWAP, for example, is a perfect entangler (it can be written as CZ$\cdot$SWAP). Their entangling action might create unwanted correlations that corrupt the computation, and be the reason why neither gate can play the role of SWAP in this context.

We will clarify this question by pinpointing precisely the characteristic that a parity-preserving gate must have in order to form a universal set together with matchgates. Before we address this question, we must turn to the characterization of two-qubit unitaries in terms of their nonlocal parameters. 

\subsection{Nonlocal parameters} \label{sec nlpar}

Let us begin this section by stating a result which will be fundamental for the discussion that follows:
\begin{theorem} \label{NonlocalTheorem}\cite{Kraus01, Khaneja01}
Any two-qubit gate $U \in SU(4)$ can be written as 
\begin{align} \label{NonlocalPar}
U & =  (U_1 \otimes U_2)  U_{NL} (V_1 \otimes V_2) \notag \\
& = (U_1 \otimes U_2) e^{i \left ( a X \otimes X + b Y \otimes Y + c Z \otimes Z \right ) } (V_1 \otimes V_2) 
\end{align}
where $U_i$ and $V_i$ are single-qubit gates on qubit $i$, and $a$, $b$ and $c$ are real parameters in the interval $[0,\frac{\pi}{2})$.
\end{theorem}

In other words, Theorem \ref{NonlocalTheorem} states that, out of the 15 parameters that define a general 4x4 unitary matrix (up to a global phase), only 3 are truly nonlocal (in the sense that they require some interaction between the two qubits), while the remaining 12 can be implemented by two sets of single-qubit gates, to the left and to the right of this nonlocal ``core''.

The set of nonlocal parameters described above is convenient because it gives an explicit way to parametrize two-qubit unitaries, but it does not seem to offer much physical insight. One reason is the non-uniqueness of the decomposition, as we can implement permutations of the type $a \leftrightarrow c$ using only single-qubit gates (in this case, $H^{\otimes 2}$). This generates an ambiguity if one tries to characterize unitaries in terms of the triplet $\{ a,b,c \}$, since any unitary gate $U$ can be associated with a (discrete) number of different sets of parameters, related to each other by permutations. One usual way of dealing with this is choosing $a \geq b \geq c$ or a similar condition, which eliminates the ambiguity by fixing one among all possible sets for a particular matrix (alternatively, for a geometrical approach in terms of what is known as the Weyl chamber, see \cite{Zhang03}). As we will see in the next sections, however, for the sets of matrices we are interested in this is not an issue, and we will make no such restriction.

While $\{ a,b,c \}$ are not, strictly speaking, local invariants, one can always write true local invariants as (symmetric) functions of these parameters. One such example is the set of invariants defined by Makhlin in \cite{Makhlin02}, as shown in \cite{Zhang03}. Another example, in which we will be particularly interested in the discussion that follows, is the entangling power defined by Zanardi \textit{et al.} in \cite{Zanardi00}.

The entangling power $e_p(U)$ of an unitary $U$ is defined as the average entanglement generated by the action of $U$ on all product states $\ket{\psi_1} \otimes \ket{\psi_2}$ \cite{Zanardi00}:
\begin{equation}
e_p(U) = \overline{E(U \ket{\psi_1} \otimes \ket{\psi_2})}^{(\psi_1, \psi_2)},
\end{equation}
where the bar denotes average with respect to some probability distribution $p(\psi_1, \psi_2)$. Simple manipulations show that this is equal to
\begin{equation}
e_p(U) = 2 \; \textrm{tr} \left[ U^{\otimes 2} \Omega_p U^{\dagger \otimes 2} P_{13}^{-} \right ], 
\end{equation}
where $P_{13}^{-}$ is the projector on the antisymmetric subspace of qubits 1 and 3 (notice that the trace is taken over a doubled Hilbert space), while
\begin{equation*}
\Omega_p = \int d \mu(\psi_1, \psi_2) (\ket{\psi_1} \bra{\psi_1} \otimes \bra{\psi_2} \ket{\psi_2})^{\otimes 2},
\end{equation*}
and $d \mu$ denotes the measure over the space of product states induced by probability distribution $p(\psi_1,\psi_2)$.

It can be shown that, if the average is taken over the uniform distribution, the entangling power is both local invariant and SWAP invariant (that is, it remains the same if $U$ is multiplied by the SWAP or by single-qubit gates on either side), which, as previously mentioned, suggests that it must be written as some function of $\{ a,b,c \}$. If fact, simply by plugging Eq. (\ref{NonlocalPar}) in the above definition (and rescaling the value of $e_p(U)$ so it goes from 0, for local gates and SWAP, up to 1, for perfect entanglers) we find that
\begin{equation} \label{entanglement}
e_p(U) = 1 - \textrm{cos}^2 2a \; \textrm{cos}^2 2b \; \textrm{cos}^2 2c - \textrm{sin}^2 2a \; \textrm{sin}^2 2b \; \textrm{sin}^2 2c.
\end{equation}

It is clear from the expression above that any local gate has $e_p(U_1 \otimes U_2) = 0$. It can also be seen that $e_p(CNOT)=1$, since for the CNOT $\{ a,b,c \}=\{\frac{\pi}{4}, 0, 0\}$, and that $e_p($SWAP$)=0$, since for the SWAP we have $\{\frac{\pi}{4}, \frac{\pi}{4}, \frac{\pi}{4}\}$, all results which were expected. We also see that the above expression is invariant under permutations of $\{a,b,c\}$, another previously anticipated feature.

Recall that the set of single-qubit gates together with any entangling gate is universal for quantum computation. We can check that a gate is entangling by calculating the non-local invariant defined above. Can a similar characterization be given in the context of matchgates? The SWAP gate, which seems to be responsible for boosting their computational power, cannot be implemented only by single-qubit gates, and thus must clearly have some nonlocal property (even though it is not entangling). We will show in the next section that this property can be characterized in terms of the $\{ a,b,c \}$ parameters above, and that it is fundamental for granting computational power to matchgates.

\section{Extending matchgates with parity-preserving unitaries}

In this section, we present our main result, regarding which parity-preserving (P.P.) unitaries can be added to the set of nearest-neighbor matchgates in order to create a universal set. In section \ref{sec iswap}, we define a continuous family of P.P. gates that interpolates between the SWAP and iSWAP, in order to understand why SWAP can boost nearest-neighbor matchgates to quantum universality whereas iSWAP cannot. In section \ref{sec matchnlpar} we characterize an arbitrary P.P. gate in terms of its nonlocal parameters, identifying both the subgroup of matchgates and the continuous family of section \ref{sec iswap} in this context. Finally, in section \ref{sec main} we prove the main theorem by determining under which conditions a specific P.P. gate can replace the SWAP in the universal set of section \ref{sec pp}, and explicitly constructing a universal set whenever these conditions are satisfied.

\subsection{Interpolating between SWAP and iSWAP} \label{sec iswap}

Let us begin by addressing a question raised at the end of section \ref{sec pp}: why is it that, in the context of nearest-neighbor matchgates, the SWAP gate is sufficient to generate a universal set, but other gates similar to it (such as iSWAP or fermionic SWAP) are not? In order to answer this question, let us replace the SWAP by a continuous family that interpolates between it and iSWAP, given by
\begin{equation}
G(I, e^{i \tau} X) =
\left(\begin{array}{cccc}
1 & 0 & 0 & 0 \\
0 & 0 & e^{i \tau} & 0 \\
0 & e^{i \tau} & 0 & 0 \\
0 & 0 & 0 & 1
\end{array}\right),
\end{equation}
where $\tau$ goes from 0 (SWAP) to $\pi / 2$ (iSWAP). This is clearly a P.P. gate [cf. Eq. (\ref{PPgates})], but it is also a matchgate whenever $\tau$ = $\pi/2$ (i.e. only in the case of the iSWAP). Thus, we expect the gate that arises from this substitution in the circuit of Fig. \ref{JozsaDem}b to complete the universal set (together with the single-qubit gates from Fig. \ref{JozsaDem}a) for one extreme value of $\tau$, and to fail to do so for the other extreme. What about intermediate values? In those cases, the resulting gate implemented in the logical qubits is given by $G(H,H) G(X,X) G(I, e^{i \tau} X) G(H,H)$ or, explicitly, by:
\begin{equation}
\left(\begin{array}{cccc}
1 & 0 & 0 & 0 \\
0 & e^{i \tau} & 0 & 0 \\
0 & 0 & e^{i \tau} & 0 \\
0 & 0 & 0 & -1
\end{array}\right).
\end{equation}

This matrix is diagonal, and thus preserves the encoding [cf. Eq. (\ref{encoding})], which is a requirement for it to be used in an encoded universal set. Simple calculations show us that the three nonlocal parameters [cf. Eq. (\ref{NonlocalPar})] for this gate are $\{ 0 , 0, \frac{\pi}{4} - \frac{\tau}{2} \}$ and so, by Eq. (\ref{entanglement}), we find that its entangling power is $\cos^2 \tau$. This confirms our previous intuition: the entangling power is maximum when we use the SWAP, but goes to 0 as the gate used approaches the iSWAP, in which case any circuit composed of these gates must become classically simulatable. The simulability of the resulting circuit can be understood in two ways: either by remembering that, in this case, the whole circuit is composed of matchgates, or by noticing that, when the two-qubit gate creates no entanglement, any circuit acting on the logical qubits can be expressed in terms only of single-qubit gates.

However, perfect entanglers are not necessary to form a universal set \cite{Bremner02}. If arbitrary single-qubit gates are available, a gate that creates any (nonzero) amount of entanglement is sufficient for universal quantum computation, with at most a polynomial overhead with respect to the set with perfect entanglers. In view of this, we see that the $G(I, e^{i \tau} X)$ gate above is sufficient to achieve universal quantum computation together with matchgates if and only if $\tau \neq \pi/2$ - or, equivalently, \textit{if and only if it is not a matchgate}.

This shows that there is a continuous family of gates akin to the SWAP which can be used to create a universal set. Furthermore, the inability of iSWAP or fermionic SWAP to provide the extra computational power cannot be attributed to the fact that they are also entanglers, since the family $G(I, e^{i \tau} X)$ includes members with all possible values of entangling power. While it is true that the most efficient gates in this set are the least entangling (which is evident by the particular cases of the SWAP and iSWAP), it will become clear in the next sections that this is just an attribute of this particular family, and there is no such relation for P.P. gates in general.

This shows that the parameter $\tau$ in the parametrization above seems fundamental for the gate's ability to generate a universal set together with matchgates. In what follows, we shall relate $\tau$ to the nonlocal parameters defined in the previous section and to the relative phase between the submatrices of a P.P. gate (section \ref{sec pp}), investigating its importance to the computational power of general P.P. matrices when associated with matchgates.

\subsection{Nonlocal parameters of P.P. matrices} \label{sec matchnlpar}

Let us start by obtaining a convenient parametrization of general P.P. matrices. To this end, let us recall that any single-qubit unitary matrix can be parametrized as
\begin{equation*}
\left(\begin{array}{cc}
\cos \theta e^{i (\beta + \alpha )} & i \sin{\theta} e^{i (\beta + \gamma )} \\
i \sin{\theta} e^{i (\beta - \gamma )} & \cos{\theta} e^{i (\beta - \alpha )}
\end{array}\right),
\end{equation*}
where $\theta$, $\beta$, $\alpha$ and $\gamma$ are real parameters in the interval $[0,2 \pi ]$. Notice that the determinant of this matrix, $e^{2 i \beta}$, depends only on $\beta$, which will be crucial for the following discussion. By parameterizing $A$ and $B$ this way, the most general P.P. matrix can be written as 
\begin{widetext}
\begin{equation} \label{GenPP}
G(A,B) = \left(\begin{array}{cccc}
\cos \theta \; e^{i (\beta + \alpha )} & 0 & 0 & i \sin{\theta} \; e^{i (\beta + \mu )} \\
0 & \cos \phi \; e^{i (-\beta + \gamma )} & i \sin{\phi} \; e^{i (-\beta + \nu )} & 0 \\
0 & i \sin{\phi} \; e^{i (-\beta - \nu )} & \cos \phi \; e^{i (-\beta - \gamma )} & 0 \\
i \sin{\theta} \; e^{i (\beta - \mu )} & 0 & 0 & \cos \theta \; e^{i (\beta - \alpha )}
\end{array}\right)
\end{equation}
in terms of 7 real parameters: $\{ \theta, \alpha, \gamma, \phi, \mu, \nu, \beta \}$. Notice that, while $A$ and $B$ are each described by 4 parameters, $G(A,B)$ is defined only up to global phase, which means we can choose the determinant parameter suitably such that det($A$) = det($B$)$^{-1}$ = $e^{2 i \beta}$.

How much nonlocal freedom remains in the 7-parameter family above? In other words, how does the P.P. restriction constrain the nonlocal parameters of $G(A,B)$? To answer this, let us explicitly write the nonlocal ``core'' defined in Eq. (\ref{NonlocalPar}):
\begin{equation} \label{NonlocalCore}
U_{NL} =
\left(\begin{array}{cccc}
\cos (a-b) \; e^{i c} & 0 & 0 & i \sin{(a-b)} \; e^{i c} \\
0 & \cos (a+b) \; e^{-i c} & i \sin{(a+b)} \; e^{-i c} & 0 \\
0 & i \sin{(a+b)} \; e^{-i c} & \cos{(a+b)} \; e^{-i c} & 0 \\
i \sin{(a-b)} \; e^{i c} & 0 & 0 & \cos (a-b) e^{i c}
\end{array}\right).
\end{equation}
\end{widetext}

Comparing Eq. (\ref{GenPP}) and Eq. (\ref{NonlocalCore}), we see that $U_{NL}$ is, itself, a P.P. gate. Not only that, but simple inspection shows us that the nonlocal parameters of $G(A,B)$ are given by $\{ \frac{\theta + \phi}{2}, \frac{\phi - \theta}{2}, \beta \}$, which points us to two important facts. First, that general P.P. gates can have any value of the nonlocal parameters (and, consequently, of any local invariant derived from them), which means that any matrix in U(4) is locally equivalent to a P.P. gate. Second, that the matchgate condition ($\beta=0$) is equivalent to fixing one of these nonlocal parameters ($c=0$).

An important remark must be made here, regarding the previously mentioned ambiguity of the triplet $\{ a,b,c \}$. In the previous section, we mentioned that permutations between these parameters could be implemented using only local gates, which meant these parameters were not local invariants \textit{per se} (and any true invariant, such as the entangling power in Eq. (\ref{entanglement}), must be symmetric under these permutations). Notice, however, that the local gates which implement these permutations are, in general, not parity-preserving (the Hadamard gate, for example). 
In fact, as will become clear below, all single-qubit P.P. gates commute with the $Z\otimes Z$ gate. This means that the P.P. condition distinguishes parameter $c$ from the other two in a qualitative way, by making it a true local invariant (and so our characterization of matchgates as the particular case of P.P. gates with $c=0$ is meaningful).

Now compare Eq. (\ref{GenPP}) and Eq. (\ref{NonlocalCore}) again. Three of the independent parameters of $G(A,B)$ have been accounted for, in terms of $\{ a,b,c \}$. The remaining four parameters are given by individual phases, and can be achieved by multiplying whole rows and columns by phases (i.e., by applying $Z$-rotations to the left and to the right). In fact, a simple calculation shows that
\begin{equation} \label{PPdecomposition}
G(A,B) = \left[ R_z(\chi _1) \otimes R_z(\chi _2) \right] U_{NL} \left[ R_z(\xi _1) \otimes R_z(\xi _2) \right ],
\end{equation}
with an appropriate choice of $\chi_1$, $\chi_2$, $\xi_1$ and $\xi_2$ in terms of $\alpha$, $\gamma$, $\mu$ and $\nu$. Now, notice that these $Z$-rotations can be rewritten as 
\begin{align*}
R_z(\chi _1) \otimes R_z(\chi _2) = {} & \textrm{diag} \left ( e^{i(\chi _1 + \chi _2)}, e^{i(\chi _1 - \chi _2)},  \right. \notag \\
& \qquad \; \left. e^{i(\chi _2 - \chi _1)}, e^{i(- \chi _1 - \chi _2)} \right ) \notag \\
= {} & G \bigl ( R_z(\chi _1 + \chi _2), R_z(\chi _1 - \chi _2) \bigr ),
\end{align*}
which \textit{is a matchgate}. Together with Eq. (\ref{PPdecomposition}), this leads to a simple result which will be important later on:
\begin{lemma}
Any P.P. gate can be written as
\begin{equation} \tag{\ref{PPdecomposition}$^\prime$}
\begin{split}
G(A,B) = {}& G \bigl ( R_z(\tau_1), R_z(\tau_2) \bigr ) e^{ i(a X \otimes X + b Y \otimes Y + c Z \otimes Z)} \\
& \times G \bigl ( R_z(\tau_3), R_z(\tau_4) \bigr ) .
\end{split}
\end{equation}
\end{lemma}

The importance of this result is that it offers a decomposition for P.P. gates in terms of matchgates and a single nonmatchgate P.P. gate. Since we are interested in the most general P.P. gate with the ability to enhance the computational power of matchgates, we will be able to disregard the $Z$-rotations in the above decomposition. Being matchgates themselves, they cannot offer any computational power beyond that of the whole matchgate group, and it will be sufficient to consider only the nonlocal term in the middle [i.e. $U_{NL}$ in Eq.(\ref{NonlocalPar})].

\subsection{Main result} \label{sec main}

We are now in condition to generalize the result of section \ref{sec iswap}, by searching for the most general $G(A,B)$ gate to replace the SWAP alongside matchgates in the universal set. First, let us consider how far we can generalize the circuit in Fig. \ref{JozsaDem}b. 

In view of the encoding given in Eq. (\ref{encoding}) we see that, in order to consider gates acting only on the pair $(2, 3)$, the resulting gate must be diagonal. Any nondiagonal component would change the parity of one of the pairs $(1,2)$ or $(3,4)$, thus disrupting the encoding. By the property that $G(A,B)$ is block diagonal, we see that, if both $A$ and $B$ satisfy
\begin{align*}
H \; X \; A \; H & = e^{i c} R_z(\sigma_1), \\
H \; X \; B \; H & = e^{- i c} R_z(\sigma_2),
\end{align*}
then the resulting gate is
\begin{equation*}
G \bigl( e^{ic} R_z(\sigma_1 ), e^{- i c} R_z(\sigma_2) \bigr ),
\end{equation*}
which is the most general diagonal gate.

This means that both $A$ and $B$ must be, at most, $X$-rotations (multiplied by a phase), and it is clear that the family considered previously (which interpolates between SWAP and iSWAP) is a particular case of this. The most general case, thus, is a matrix of the form $G \bigl (e^{i c} R_x(\theta), e^{-ic} R_x(\phi) \bigr )$, given explicitly by
\begin{equation} \label{GeneralUsefulPP}
\left(\begin{array}{cccc}
\cos (\theta) \; e^{i c} & 0 & 0 & i \sin{(\theta)} \; e^{i c} \\
0 & \cos (\phi) \; e^{-i c} & i \sin{(\phi)} \; e^{-i c} & 0 \\
0 & i \sin{(\phi)} \; e^{-i c} & \cos (\phi) \; e^{-i c} & 0 \\
i \sin{(\theta)} \; e^{i c} & 0 & 0 & \cos (\theta) \; e^{i c}
\end{array}\right),
\end{equation}
which has the exact form of the nonlocal gate $U_{NL}$, if we take $a-b = \theta$ and $a+b = \phi$. 

Thus, we conclude that the most general gate which can replace the SWAP in the circuit in Fig. \ref{JozsaDem}b while still preserving the encoding is of the form $U_{NL} = \textrm{exp} \left[ i(a X \otimes X + b Y \otimes Y + c Z \otimes Z) \right ]$ which, by the remark at the end of the previous section, also happens to be the most general P.P. gate we \textit{need} to consider in order to characterize the ability of \textit{any} P.P. gate to boost the computational power of matchgates.

By replacing the SWAP with the gate in Eq. (\ref{GeneralUsefulPP}) above, the resulting gate, $G(H,H) G(X,X) U_{NL}$ $\times \: G(H,H)$, is given explicitly by (in terms of $\{ a, b, c \}$):
\begin{equation}
\left( \begin{array}{cccc}
e^{i (a - b + c)} & 0 & 0 & 0 \\
0 & e^{i (a + b - c)} & 0 & 0 \\
0 & 0 & -e^{i (-a - b - c)} & 0 \\
0 & 0 & 0 & -e^{i (-a + b + c)})
\end{array} \right).
\end{equation}

Again, in the same fashion as in section \ref{sec iswap}, we must determine the gate's entangling power. As long as the entangling power is nonzero, this gate completes the universal set together with the single-qubit gates of Fig. \ref{JozsaDem}a. By obtaining the nonlocal parameters of the above gate and using expression (\ref{entanglement}), we find that its entangling power is $\sin^2(2c) $. Again, as before, this means that the gate creates \textit{some} entanglement as long as $c \neq 0$ - that is, as long as it is not a matchgate. We now have all ingredients for the proof of our main result:

\begin{theorem} \label{MainTheorem}
Let $G(A,B)$ be \textbf{any} parity-preserving unitary, and $M$ be the set of all matchgates + $G(A,B)$. Then, if $G(A,B)$ is not a matchgate [i.e. if det(A) $\neq $ det(B)], a universal set can be created from gates in $M$ acting only on nearest-neighbor qubits.
\end{theorem}

\begin{proof}
Let us give a constructive proof by supposing one has the ability to use nearest-neighbor matchgates whenever necessary, in addition to a particular $G(A,B)$ gate.

\newcounter{itemcounter}
\begin{list}{\textbf{\arabic{itemcounter}.}}{\usecounter{itemcounter}\leftmargin=1.0em}
\item Encode logical qubits as $\ket{0}_L = \ket{00}$ and $\ket{1}_L = \ket{11}$;
\item Obtain the parameters for matrix $G(A,B)$ such that it is written in the form
\begin{equation*}
R_z(\chi _1) \otimes R_z(\chi _2) e^{i (a X \otimes X + b Y \otimes Y + c Z \otimes Z)}  R_z(\xi _1) \otimes R_z(\xi _2);
\end{equation*}
\item Use the circuit in Fig. \ref{JozsaDem}a to implement any single-qubit gate;
\item Modify the circuit of Fig. \ref{JozsaDem}b such that the $Z$ rotations cancel out, and all that is left is 
\begin{equation*}
G(H,H) G(X,X) e^{i (a X \otimes X + b Y \otimes Y + c Z \otimes Z)} G(H,H);
\end{equation*}
this is an allowed procedure, since the gates required to modify the circuit are matchgates themselves. The final circuit can be seen in Fig. \ref{OurDem}.
\item The modified circuit now implements the unitary
\begin{equation*} 
\left( \begin{array}{cccc}
e^{i (a - b + c)} & 0 & 0 & 0 \\
0 & e^{i (a + b - c)} & 0 & 0 \\
0 & 0 & -e^{i (-a - b - c)} & 0 \\
0 & 0 & 0 & -e^{i (-a + b + c)})
\end{array} \right) ,
\end{equation*}
which generates entanglement $\sin^2(2c)$. By a known result \cite{Bremner02}, we know that any entangling matrix, along with arbitrary single-qubit gates, is universal for quantum computation. So, the gate implemented by the circuit in Fig. \ref{OurDem} completes the universal set as long as $c \neq 0$ or, in other words, as long as it is not a matchgate.
\end{list}

\begin{figure}
\includegraphics[width=3in]{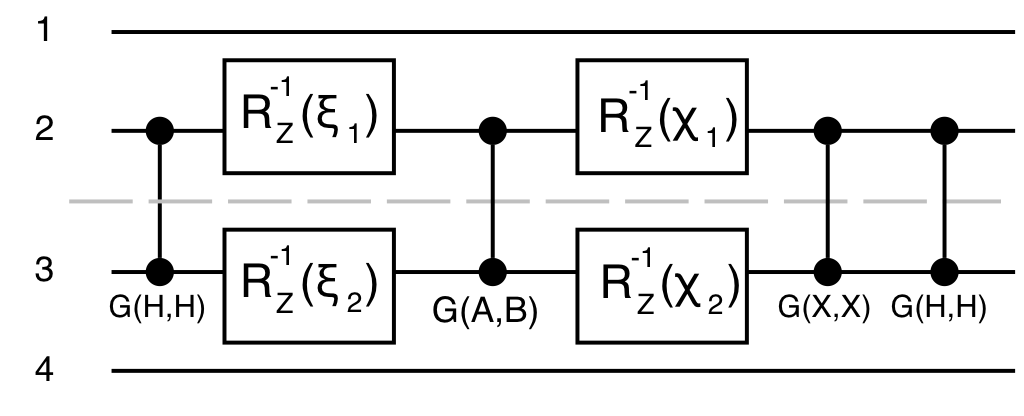}
\caption{Circuit where the SWAP is replaced by G(A,B) and some additional matchgates.}
\label{OurDem}
\end{figure} 

\end{proof}

Note that this result can be complemented by the well-known classical simulability of matchgate circuits \cite{Terhal02}. This means that, at least when restricted to parity-preserving nearest-neighbor operations, the jump from (sub-)classical computation to full universal quantum computation is abrupt, in the sense that either the set $M$ of matchgates + $G(A,B)$ is universal for quantum computation or it generates only classically simulatable circuits. This did not need to be so, as there are recent results about families of quantum circuits which are expected to be neither universal for quantum computation nor classically simulatable. Examples include the so-called IQP circuits \cite{Bremner10} and linear optics without adaptive measurements  \cite{Aaronson11}.

\section{Discussion} \label{sec disc}

So far, we have restricted our attention to P.P. gates, as they are a natural generalization of matchgates. Which other SU(4) unitaries can, together with matchgates, form a universal set? It is easy to see that matchgates + arbitrary single-qubit gates form a universal set, since matchgates include many perfect entanglers \cite{Terhal02}. We also know that the Hadamard gate alone is sufficient for this purpose since, as mentioned in section \ref{sec nlpar}, conjugating a P.P. unitary by $H^{\otimes 2}$ implements the permutation $a \leftrightarrow c$. Thus we could take a matchgate with $a \neq 0$ and, using the $H$ gate, construct a gate with $c \neq 0$ satisfying the conditions of Theorem \ref{MainTheorem}. Clearly, other $H$-like gates which implement similar permutations (i.e. $a \leftrightarrow c$ and $b \leftrightarrow c$, but not $a \leftrightarrow b$) can be used in the same way. It can also be shown easily that such gates are the only single-qubit gates which take matchgates into general P.P. gates by conjugation. While the ability of $H$ to extend the power of matchgates appears very naturally in the present construction, it is not particularly surprising: matchgates contain all single-qubit $Z$-rotations which, when added to the $H$ gate, can generate arbitrary single-qubit gates, completing a universal set even without encoding.

Our results used minor modifications of the encoding proposed in \cite{Jozsa08b} and rely heavily on it. Since we encoded each logical qubit in the parity of the physical qubits, we considered only P.P. gates or, alternatively, operations which are not P.P. themselves, but which preserve a certain subset of P.P. unitaries (i.e. conjugation by $H^{\otimes 2}$). It is an open question whether different encodings may help identify different unitaries which generate a universal set together with matchgates.

Another question we may ask is whether there are other interesting subsets of P.P. gates which can be shown to be either classically simulatable or quantum universal. Since matchgates are P.P. gates with $c=0$, two other natural subsets to consider would be P.P. gates with $a=0$ or $b=0$. Let us consider this latter case, and denote by $M_b$ the subset of all P.P. gates with $b=0$. Notice that, in Figs. \ref{JozsaDem}a and \ref{OurDem}, only $G(A,A)$ matrices are used, along with an arbitrary $G(A,B)$ and local $Z$-rotations. Inspecting Eq. (\ref{GenPP}) we see that, in addition to having $c=0$, $G(A,A)$ has $b=0$ (i.e. $\phi = \theta$). This means that the present construction already uses mostly gates from $M_b$. By choosing $G(A,B)$ as any gate from $M_b$ with $c \neq 0$, we satisfy the conditions of Theorem \ref{MainTheorem}, and thus complete a universal set. This proves that P.P. gates with $b=0$, or equivalently with $a=0$, are already universal for quantum computation, in sharp contrast with the $c=0$ case. Thus, the theorem actually ascribes a special status for parameter $c$, which we already anticipated in section \ref{sec matchnlpar} by showing it was qualitatively different from the others.

The above result can be understood easily in the context of the fermionic nature of matchgates. As can be seen in \cite{Terhal02, Jozsa08b}, matchgates are the group generated by the set $\{ X \otimes X$, $Y \otimes Y$, $1 \otimes Z$, $Z \otimes 1$, $X \otimes Y$, $Y \otimes X \}$. Using the Jordan-Wigner transformation, these matrices are shown to be equivalent to Hamiltonians quadratic in creation and annihilation operators, describing the evolution of noninteracting fermions. On the other hand, the $Z \otimes Z$ matrix, related to parameter $c$ in our construction, is equivalent to a Hamiltonian quartic in creation/annihilation operators, and so it describes an interaction. Thus, the break in the permutation symmetry imposed by the P.P. condition, which is not obvious in the algebraic approach, turns out to have a very physical meaning in the fermionic approach.

In fact, in \cite{Bravyi02}, Bravyi and Kitaev have shown that the $exp \{ i \frac{\pi}{4} Z \otimes Z \}$ interaction can be used to implement a $CZ$, which is then used to generate a universal set with fermionic modes. Our result is similar in spirit to this one, but also more general as we provide an explicit construction where \textit{any} nonmatchgate $G(A,B)$ matrix is sufficient to generate a universal set. This alternative approach can be useful if one is attempting to implement circuits of P.P. gates in other physical systems, rather than fermions. In this case, it is not obvious what property gates such as $CZ$ and SWAP have in common, except for their ability to provide matchgates with full quantum computational power. The power of the SWAP is attributed to the long-range interaction it simulates, whereas the CZ is used to generate entanglement between the encoded qubits (see Fig. \ref{JozsaDem}b). Thus, Theorem \ref{MainTheorem} connects both results by identifying the property these gates share, showing that it is precisely the resource missing in the group of matchgates and providing an explicit construction for any other gate of this kind. 

Other interesting subsets of P.P. gates for which similar results are known are the unitary evolutions generated by Heisenberg isotropic ($\mathcal{H}_i= X \otimes X + Y \otimes Y + Z \otimes Z$) or anisotropic ($\mathcal{H}_a = X \otimes X + Y \otimes Y$) interactions \cite{Kempe01a, Kempe01b}. These sets are of great physical interest, as they provide examples where a quantum computer could be built out of a single interaction Hamiltonian. Explicit constructions for said universal sets can be found, respectively, in \cite{DiVincenzo00} and \cite{Kempe02}. One fundamental difference between these results is that the isotropic interaction considered in \cite{DiVincenzo00} acts only on nearest-neighbors, while the anisotropic interaction in \cite{Kempe02} acts on next-nearest-neighbors. Furthermore, it has been shown that the anisotropic interaction is not universal while acting only on nearest-neighbors, meaning that distant interactions are actually necessary in this case \cite{Wu02}. These two results are clearly consistent with our own, as any unitary evolution generated by $\mathcal{H}_a$ must be a matchgate, while $\mathcal{H}_i$ generates nonmatchgate P.P. unitaries. It is important to point out, however, that neither result implies or is implied by our own.

\section{Conclusion}\label{sec conclusion}

In \cite{Jozsa08b} it was shown that the SWAP gate can be used to extend the computational power of nearest-neighbor matchgates to full quantum computation. We have expanded on this result by showing that the SWAP can be replaced by any parity-preserving (P.P.) two-qubit gate which is not a matchgate itself. We have also characterized arbitrary parity-preserving gates in terms of their three nonlocal parameters $\{ a,b,c \}$, showing that matchgates form a particular subgroup of these where $c=0$. Thus, parameter $c$ is seen to be qualitatively different from the other two. This is supported by the fact that the set of P. P. gates with $b=0$ is universal for quantum computation (which is true also for gates with $a=0$).

We have also discussed the relationship between our work and other known results in the context of fermions. While matchgates are equivalent to noninteracting fermions, the parameter $c$ is known to be associated with certain interactions, providing an underlying physical interpretation to its special status. Our result, however, is more general in the sense that it does not rely on the fermionic nature of matchgates, being also suitable for implementations in other physical systems. 

We hope that the characterization of unitaries in terms of their nonlocal parameters will prove fruitful in understanding further aspects of the separation between classical and quantum computation. On a more practical note, we believe it could open the way to the construction of different universal sets of quantum operations, perhaps more adequate for particular experimental quantum computer implementations.
\begin{acknowledgments}
We would like to thank Daniel Jonathan for helpful discussions and the Brazilian funding agency CNPq (Conselho Nacional de Desenvolvimento Cient\'ifico e Tecnol\'ogico) for financial support.
\end{acknowledgments}

\bibliographystyle{unsrt}

\end{document}